\documentclass[12pt]{article}
\usepackage[utf8]{inputenc}
\usepackage{amsmath, amssymb, amsthm, mathrsfs, mathtools}
\usepackage{slashed}
\usepackage[colorlinks=true, linkcolor=black, citecolor=black, urlcolor=blue]{hyperref}
\usepackage{tikz}
\usetikzlibrary{cd}
\usepackage{graphicx}
\usetikzlibrary{decorations.markings, decorations.pathreplacing, calc, arrows.meta, positioning}
\usetikzlibrary{shapes.geometric}
\usepackage{titlesec}

\theoremstyle{plain}

\newtheorem{lemma}{Lemma}

\theoremstyle{definition}

\theoremstyle{remark}

\usepackage[left=2.5cm,right=2.5cm,top=2.5cm,bottom=3cm]{geometry}
\fontdimen2\font=1.2\fontdimen2{\jot}{5pt}

\DeclareMathOperator{\Tr}{\mathrm{Tr}}

\newcommand{\vp}{\varphi}

\newcommand{\beq}{\begin{equation}\begin{aligned}}
\newcommand{\eeq}{\end{aligned}\end{equation}}

\newcommand{\al}{\alpha}

\newcommand{\pa}{\partial}

\newcommand{\ov}{\over}
\newcommand{\ep}{\epsilon}

\newcommand{\lam}{\lambda}

\newcommand{\Christoffel}[3]{\left\{ \begin{smallmatrix} #1 \\ #2\, #3 \end{smallmatrix} \right\}}

\titleformat{\section}{\large\bfseries}{\thesection.}{4pt}{}
\titlespacing{\section}{0pt}{20pt}{6pt}

\titleformat{\subsection}{\normalfont\bfseries}{\thesubsection.}{4pt}{}
\titlespacing{\subsection}{0pt}{15pt}{6pt}

\titleformat{\subsubsection}{\normalfont\itshape}{\thesubsubsection.}{4pt}{}
\titlespacing{\subsubsection}{0pt}{15pt}{6pt}

\titleformat{\paragraph}{\normalfont\itshape}{\theparagraph.}{4pt}{}
\titlespacing{\paragraph}{0pt}{15pt}{6pt}

\DeclareFontShape{OT1}{cmr}{mx}{n}%
{<->cmr10}{}
\newcommand{\mytitlefont}{\fontseries{mx}\selectfont}
\DeclareMathAlphabet{\titlemath}{OT1}{cmr}{mx}{n}

\begin{document}

\begin{titlepage}

\begin{center}
			
~\\[0.9cm]
			
{\fontsize{24pt}{0pt} \mytitlefont  Global Anomalies in Sigma Models with Majorana--Weyl Fermions}
			
~\\[1cm]

{\fontsize{15pt}{20pt}\selectfont Changha Choi$^{a}$\,\footnote{\href{mailto:cchoi@perimeterinstitute.ca}{\tt cchoi@perimeterinstitute.ca}}}

~\\[0.5cm]

{\it $^{a}$Perimeter Institute for Theoretical Physics,\\ 
Waterloo, Ontario, N2L 2Y5, Canada}\

\end{center}

\vskip0.5cm
			
\noindent 

\makeatletter
\renewenvironment{abstract}{%
    \if@twocolumn
      \section*{\abstractname}%
    \else 
      \begin{center}%
        {\bfseries \normalsize\abstractname\vspace{\z@}}
      \end{center}%
      \quotation
    \fi}
    {\if@twocolumn\else\endquotation\fi}
\makeatother

\begin{abstract}
We investigate a global sigma model anomaly in two-dimensional sigma models with Majorana--Weyl fermions coupled to a sigma model field with target space~$G$. The anomaly originates from the nontrivial topology of the space of maps and manifests as a phase in the fermion path integral. Using the global anomaly formula expressed in terms of the reduced~$\eta$-invariant, we demonstrate that this anomaly modifies the standard quantization condition of the Wess--Zumino term on~$G$, in close analogy with the three-dimensional parity anomaly. However, our situation is more refined and highlights a qualitatively new phenomenon in two dimensions: whereas in three dimensions the anomalous quantization restricts the level to lie in a half-integer lattice, here it can force the level to take \emph{arbitrary real values}. Furthermore, our results support the consistency of the low-energy description proposed by Gaiotto, Johnson-Freyd, and Witten for three-dimensional $\mathcal{N}=1$ supersymmetric Yang--Mills theory on an interval, by highlighting a subtle and qualitatively distinct nature of the sigma-model anomalies.
\end{abstract}
\vfill
\end{titlepage}

\tableofcontents

\section{Introduction}

Anomalies are among the most fundamental pieces of kinematical data in a quantum field theory. Fermionic anomalies, in particular, are among the earliest and most illustrative examples. They arise when the fermion path integral,
\[
Z_{\text{fermion}}[A] = \int \mathcal{D}\psi \, e^{-S[\psi, A]},
\]
fails to remain invariant under gauge transformations of the background field \( A \). Under a gauge transformation \( A \mapsto A^g \), the fermion partition function may acquire a phase,
\[
Z_{\text{fermion}}[A^g] = e^{i\alpha[g, A]} Z_{\text{fermion}}[A],
\]
which cannot be removed by adding a local counterterm. This signals the presence of a quantum anomaly. If \( A \) is a dynamical gauge field, such anomalies reflect a fundamental inconsistency and must cancel. In contrast, if \( A \) is a background field for a global symmetry, the resulting ’t Hooft anomaly places nontrivial constraints on possible IR phases, dualities, and symmetry realizations.

In this standard framework, anomalies are classified as either local or global, depending on whether the gauge transformation \( g \) is continuously connected to the identity. Local anomalies arise from infinitesimal, perturbative transformations and can be detected through the non-conservation of classically conserved currents or via anomalous variations of correlation functions. In contrast, global anomalies are associated with large gauge transformations—those not connected to the identity—and typically require nonperturbative tools, such as index theorems or topology of the background field space, to detect.

In this paper, we are primarily concerned with global anomalies, which are particularly interesting because they impose subtle consistency constraints that are invisible to local, perturbative analysis. Two of the most well-known examples are Witten's SU(2) anomaly in four dimensions \cite{Witten:1982fp}, which obstructs the consistent quantization of a single Weyl fermion in the fundamental representation of SU(2), and the parity anomaly in three-dimensional non-abelian gauge theories with charged fermions, notably analyzed by Redlich \cite{Redlich:1983dv, Redlich:1983kn}, where gauge invariance is preserved only at the expense of parity. In both cases, the anomaly arises from large gauge transformations—those not connected to the identity component of the gauge group—corresponding to nontrivial elements in \( \pi_4(\mathrm{SU}(2)) = \mathbb{Z}_2 \) and \( \pi_3(G) = \mathbb{Z} \) for compact connected simple Lie group $G$, respectively.

This paper discusses a natural two-dimensional analogue of the $3d$ parity anomaly described above. At first glance, there appears to be no direct analogue of the gauge theory examples in two dimensions, since \( \pi_2(G) = 0 \) for any compact Lie group.\footnote{This statement is schematic, since the global gauge anomaly is not classified by $\pi_d(G)$ but rather by the spin bordism group $\Omega_{d+1}^{\mathrm{Spin}}(BG)$ (see \cite{Davighi:2020kok} and references there in). In two dimensions, however, still no such global gauge anomaly arises because $\Omega_{3}^{\mathrm{Spin}}(BG)=0$~\cite{Delmastro:2021otj}.} Nevertheless, we show that a conceptually related global anomaly can still arise in two dimensions, when Majorana--Weyl fermions are coupled not to gauge fields, but to \( G \)-valued sigma model fields.

The main motivation for this work comes from the paper of Gaiotto, Johnson-Freyd, and Witten \cite{Gaiotto:2019asa} (see also generalization \cite{Dedushenko:2022uay}), where they considered three-dimensional \( \mathcal{N}=1 \) supersymmetric Yang--Mills theory with a Chern--Simons term at level \( k_{\text{3d}} \), placed on the manifold \( \mathbb{R}^2 \times [0, L] \). In the limit \( g^2 L \ll 1 \), and with appropriate supersymmetric Dirichlet boundary conditions, they reasonably argued that the low-energy dynamics is governed by a two-dimensional \( (0,1) \) supersymmetric sigma model on \( \mathbb{R}^2 \) with a level \( k = k_{\text{3d}} \) Wess--Zumino term. The field content consists of a \( G \)-valued bosonic sigma model field \( g: \Sigma^2 \to G \), arising from open Wilson lines ended on an interval, together with Majorana--Weyl fermions valued in the pullback \( g^* TG \), originating from Dirichlet boundary conditions on a three-dimensional Majorana fermion in the adjoint representation.

Even though the paper primarily focuses on the case \( G = \mathrm{SU}(2) \), an apparent puzzle arises for other groups, such as \( G = \mathrm{SU}(3) \). This is because the three-dimensional parity anomaly gives a modified quantization condition for the Chern--Simons level:
\beq
k_{\text{3d}} -\frac{h^\vee}{2} \in \mathbb{Z},
\eeq
where \( h^\vee \) is the dual Coxeter number of \( G \). This leads to the well-known result that the effective Chern--Simons level can be half-integer~\cite{Witten:1999ds}. However, this familiar feature, when combined with the scenario proposed by Gaiotto, Johnson-Freyd, and Witten, leads to a puzzling tension in two dimensions: the resulting two-dimensional \( (0,1) \) sigma model inherits a Wess--Zumino term with the same half-integer level \( k = k_{\text{3d}} \).

A Wess--Zumino term with non-integer level is classically ill-defined, just as a Chern--Simons term at non-integer level fails to be gauge invariant. This raises the question of whether the apparent inconsistency of a non-integer WZ term can be resolved by a quantum anomaly in the underlying theory. In particular, we ask whether the fermionic path integral of a two-dimensional sigma model with Majorana--Weyl fermions carries an anomaly that precisely compensates the ambiguity of the non-integer WZ term. 

To understand anomalies involving sigma models, it is useful to adopt a geometrically motivated perspective \cite{Atiyah:1984tf,Moore:1984ws,Dai:1994kq,Witten:2015aba}: the fermion path integral can be interpreted as a section of a complex line bundle over the space of background fields. In gauge theories, this configuration space is the space of gauge fields modulo gauge transformations, \( \mathcal{A}/\mathcal{G} \), while in sigma models it is the space of maps from spacetime into the target, \( \mathrm{Map}(M^d, G) \). That is,
\beq
Z_{\text{fermion}} \in \Gamma\left( \mathcal{L} \to \mathcal{B} \right),
\eeq
where \( \mathcal{B} \) denotes the background field space (e.g., \( \mathcal{A}/\mathcal{G} \) or \( \mathrm{Map}(M^d, G) \)), and \( \mathcal{L} \) is the determinant (or Pfaffian) line bundle. In this framework, an anomaly corresponds to the nontriviality of \( \mathcal{L} \): local anomalies are encoded in its curvature, computable via the families index theorem, while global anomalies arise from its topology—e.g., nontrivial holonomy around closed loops in \( \mathcal{B} \). This geometric perspective unifies local and global anomalies, and connects naturally with the modern viewpoint of anomalies as invertible field theories in one higher dimension \cite{Freed:2014iua,Monnier:2019ytc}.

From this analysis, we find that the quantization condition in two dimensions is more refined than the standard odd-dimensional parity anomaly. In particular, the global anomaly depends continuously on the background connection on the target manifold, and in particular on its torsion, which can be tuned continuously. As a result, the quantization law is modified so that the coefficient~$k$ of the Wess--Zumino term can, in principle, take an arbitrary real value,
\beq
  k \in \mathbb{R}.
\eeq
In particular, this implies that the anomalous quantization law of the Wess--Zumino term applies not only to Majorana--Weyl fermions but also to Weyl fermions more generally.

The organization of the paper is as follows. In Section~\ref{sec:2}, we explain the setup of the problem. In Section~\ref{sec:EucMW}, we explain the meaning of the Euclidean path integral for Majorana--Weyl fermions. Section~\ref{sec:loc} discusses the local sigma model anomaly. In Section~\ref{sec:top}, we describe the topology of the configuration space and its natural connection to the Wess--Zumino--Witten term. Sections~\ref{sec:hol}, \ref{sec:APS1}, and~\ref{sec:APS2} present the global anomaly formula and compute the reduced $\eta$-invariant by suitably applying the APS index theorem. In Section~\ref{sec:quan}, we establish the anomalous quantization condition for the Wess--Zumino term, distinguishing between the cases where the sigma model field is non-dynamical and dynamical (Sections~\ref{sec:nondyn} and~\ref{sec:(1,0)}, respectively). Finally, in Section~\ref{sec:DF}, we show that our anomalous quantization law is sufficient to ensure that the theory is free of anomalies, in the modern sense of Dai and Freed.

\section*{Acknowledgments}  The author thanks J. Gomis, A. Tseytlin, and E. Witten for a careful reading and valuable comments on the manuscript. The author thanks M. Dedushenko, D. Gaiotto, M. Yamashita, and P. Yi for helpful discussions. The research was supported by the Perimeter Institute for Theoretical Physics. Research at Perimeter Institute is supported in part by the Government of Canada through the Department of Innovation, Science and Economic Development and by the Province of Ontario through the Ministry of Colleges and Universities.

\section{Global Sigma Model Anomalies in Two Dimensions} \label{sec:2}

We consider Majorana--Weyl fermions on a two-dimensional Lorentzian manifold \( \Sigma^2 \), where such spinors are well-defined. Given a Dirac spinor \( \psi \), a Majorana--Weyl fermion is one that satisfies both the Majorana condition \( \psi^T C = \bar\psi \), which imposes a reality constraint using the charge conjugation matrix \( C \), and the Weyl condition \( \gamma_c \psi = \pm \psi \), which projects onto definite chirality. In two dimensions, these constraints are compatible and reduce the spinor to a single real component.

We couple the Majorana--Weyl fermion to a (background) sigma model field \( g: \Sigma^2 \to G \), where the target is a compact Lie group \( G \). The fermion is taken to be valued in the pullback bundle \( g^* TG \), the pullback of the tangent bundle of \( G \) along the map \( g \). In terms of the globally defined (left-invariant) coframe fields \( \theta^a \) on $G$, defined via the Maurer--Cartan 1-form \( g^{-1} dg = \theta^a T_a  \) (see \cite{Choi:2021yuz} for details), the spinor field acquires an additional index corresponding to the adjoint representation of \( G \). That is, we write \( \psi^a \) with \( a = 1, \dots, \dim(G) \), where the index \( a \) labels components in the vielbein basis \( \theta^a \).

To define the Dirac operator, we must choose an affine connection $\nabla$ on the tangent bundle \( TG \). A standard choice is the Levi-Civita connection, which is torsion-free but has nonzero (constant) curvature. Alternatively, one can use a (left- or right-) parallelizing connection, which has nonzero torsion but vanishing curvature. The latter choice is particularly natural in the context of supersymmetric sigma models, including the minimal \((0,1)\) supersymmetric Wess--Zumino--Witten (WZW) model (see~\cite{Murthy:2025ioh} for recent discussions of the \((1,1)\) case).

Throughout the analysis, we take  \( \nabla \)   to be the Cartan connections which is a one-parameter family of left-invariant connection associated to the bi-invariant metric on \( G \). Physically, this arises naturally in the context of $(0,1)$ sigma models of our interest.  The defining property of the Cartan connections is (following \cite{postnikov2013geometry})
\beq\label{eq:CartanC}
\nabla_X Y = \lambda [X, Y],
\eeq
where \( X, Y \in \mathfrak{g} \) are arbitrary left-invariant vector fields on $G$ and $\lam \in \mathbb R$ represents the one-parameter family. 

The corresponding torsion and curvature tensor are given by
\beq
T(X, Y) =  (2\lambda - 1)[X, Y],\quad R(X, Y)Z = (\lam^2-\lam) [[X,Y],Z].
\eeq
There are three important special cases:
\[
\begin{dcases}
\lambda = 0 & \nabla_{{L}}:~ \text{Left Cartan connection}, \\
\lambda = \tfrac{1}{2} & \nabla_{{LC}}:~  \text{Levi--Civita connection}, \\
\lambda = 1 & \nabla_{{R}}:~  \text{Right Cartan connection}.
\end{dcases}
\]
where $\nabla_{{LC}}$ is the standard torsion-free Levi-Civita connection and $\nabla_{{L}}$ (or $\nabla_{{R}}$) is a torsionful flat connection which parallelizes $G$ w.r.t. left (or right) translations on $G$.

If we choose the basis of the (co)frame bundle to be  \( \{\theta^a\} \), the assocciated spin connection 1-form simply becomes
\beq \label{eq:MCw}
\omega^a_{\ b} = -\lambda f^a_{\ bc} \theta^c,
\eeq
where $[T_a,T_b]=f^c_{\ ab}T_c$. For the above special cases, we will denote the associated spin connections to be $w_{LC},w_{L,R}$.

Once the connection on the target $G$ is chosen, the Dirac operator for a fixed sigma model field background \( g \) is explicitly given by
\beq
\slashed{D}_g = \gamma^I u_I^\mu \left( \partial_\mu + \varpi_\mu + \mathcal A _\mu \right),
\eeq
where \( u_I^\mu \) and \( \varpi_\mu \) denote the inverse vielbein and spin connection on \( \Sigma^2 \), respectively, and \( \mathcal  A_\mu \equiv (g^* w)_\mu \) is the pullback of the \( \mathrm{SO}(\dim G) \)-valued connection on \( TG \). In terms of the trivialization \( \theta^a = \theta^a_i\, d\varphi^i \) of \( T^*G \), where \( \{ \varphi^i \} \) are local coordinates on \( G \), the pullback connection acts on the spinor \( \psi^a \) as
\beq
(\mathcal  A_\mu \psi)^a = \partial_\mu \varphi^i \, \theta_i^b \, (\mathcal  A_\mu)^a_{\;\,b} \, \psi^b.
\eeq
and using this one can explicitly write down the fermionic action.

\subsection{Euclidean Path Integral for Majorana-Weyl fermions} \label{sec:EucMW}
Although Majorana–Weyl fermions in two dimensions are strictly defined in Lorentzian signature, we follow the standard approach of analyzing anomalies via the Euclidean path integral, employing an appropriate analytic continuation. In this framework, the Wick-rotated path integral of a Majorana–Weyl fermion is naturally identified with the square root of the Euclidean path integral of a Weyl fermion~\cite{Freed:1986hv}. This reflects the fact that two Majorana–Weyl fermions in Lorentzian signature are equivalent to a single Weyl fermion.\footnote{Weyl fermions exist in even dimensions in both Lorentzian and Euclidean signature, though the Majorana condition depends sensitively on the signature.}

The Euclidean path integral of a Weyl fermion lacks a canonical definition of the determinant, since the chiral Dirac operator $D_+$ maps the space of positive chirality spinors to that of negative chirality spinors, \(\mathcal{H}^+ \rightarrow \mathcal{H}^-\) (or vice versa, depending on convention).
As a result, the path integral associated with the chiral Dirac operator is naturally valued in a complex line bundle, known as the \emph{determinant line bundle}~\cite{quillen1985determinants, bismut1986analysis} which we formally write as $\det D_+$. This complex-valuedness of the path integral reflects the fact that Weyl fermions generically suffer from perturbative anomalies.

We briefly comment that in the case of a Dirac fermion coupled to a real vector bundle \( E \rightarrow \Sigma^2 \) (as is required for consistently coupling to Majorana--Weyl fermions), the Euclidean path integral 
\beq
\det(\slashed{D}) = \det(D_- D_+)
\eeq
has a strictly positive determinant. As a result, the determinant line bundle is trivial. This positivity follows from the existence of a charge conjugation symmetry, which ensures that the spectrum of the Dirac operator is doubly degenerate. In particular, the square root of the full Dirac determinant is unambiguously defined, which reflects the fact that the absolute value of the chiral determinant is given by
\beq
|\det D_+| = \det(\slashed{D})^{1/2}.
\eeq

In summary, the Majorana--Weyl fermion path integral defines a complex line bundle of natural interest: it is the square root of the determinant line bundle of the chiral Dirac operator, commonly referred to as the Pfaffian line bundle~\cite{Freed:1986hv,Witten:1999eg}:
\beq
\mathcal{L}_{\text{MW}} = (\det D_+)^{1/2}.
\eeq
As discussed in the introduction, a sigma model anomaly reflects an obstruction to globally defining a smooth, nonvanishing section of the fermion path integral line bundle \( \mathcal{L}_{\text{MW}} \) over the space of background fields. Local anomalies arise from its curvature, while global anomalies are captured by nontrivial holonomy around loops in field space.

\subsection{Local Sigma Model Anomaly} \label{sec:loc}
Before turning to global anomalies, one must first verify that the local anomaly vanishes. In the case of sigma models \cite{Manohar:1984zj,Moore:1984ws,Alvarez-Gaume:1985bbj,Bagger:1985pw}, this involves studying a family of chiral fermions coupled to background fields that vary smoothly over a parameter space. The fermion partition function then defines a complex line bundle over this space, and any local anomaly is encoded in the curvature of this bundle. The families index theorem \cite{Atiyah:1970ws,Atiyah:1984tf} provides a concrete way to compute this curvature via fiber integration of characteristic classes constructed from the geometry of the background configuration.

Let \( Y \) be a compact, oriented 2-dimensional parameter space (e.g., \( S^2 \)), and let \( X \) be a fixed compact spin manifold of even dimension \( \dim X = 2n \). Consider a smooth family of sigma model fields \( \phi_y: X \to M \), labeled by \( y \in Y \), which we assemble into a single smooth map
\beq
\phi: Y \times X \to M.
\eeq
Let \( E \to M \) be a complex vector bundle with connection, and let \( \mathcal{V} = \phi^* E \to Y \times X \) be the pullback bundle over the total space. This defines a family of chiral Dirac operators \( \slashed{D}_y \) on \( X \), twisted by \( \phi_y^* E \), forming a family of Fredholm operators over \( Y \). The associated determinant line bundle \( \mathcal{L} \to Y \) carries a natural connection, and its curvature is computed via the  families index theorem:
\beq
\mathrm{curv}(\mathcal{L}) = {1\ov 2}\cdot 2\pi i  \left[ \int_X \hat{A}(TX) \wedge \phi^* \operatorname{ch}(E) \right]_2,
\eeq
where \( [\cdot]_2 \) denotes the degree-two component of the resulting form on \( Y \) after integrating over the fiber \( X \) and we have an overall factor of ${1\ov 2}$ since we are dealing with Majorana-Weyl fermion.

Now we specialize to our case where \( X = \Sigma^2 \) is a two-dimensional spacetime, \( M = G \) is a Lie group target, and \( E= TG \to G \) is the tangent bundle equipped with a spin connection $ w$. The sigma model field is a map \( \hat{g}: Y \times \Sigma^2 \to G \), where \( Y \) is a smooth parameter space of sigma model configurations. The pullback bundle \( \mathcal{V} = \hat{g}^* TG \to Y \times \Sigma^2 \) carries the induced connection from \( TG \), and the fermions couple to \( \mathcal{V} \) over each copy of \( \Sigma^2 \). The families index theorem then computes the curvature 2-form of the determinant line bundle \( \mathcal{L} \to Y \), which governs the local (perturbative) anomaly across the family of backgrounds. Since \( \Sigma^2 \) is two-dimensional, the \( \hat{A} \)-genus reduces to 1, and the term in \( \operatorname{ch}(TG) \) that contributes to the curvature is \( \operatorname{ch}_2 \). We therefore have (where $ R=dw +w\wedge w$ is a curvature 2-form on the target)
\begin{equation} \label{eq:curv}
\mathrm{curv}(\mathcal{L}) = \frac{1}{2} \cdot 2\pi i \int_{\Sigma^2} \hat{g}^* \operatorname{ch}_2(TG)=-{i\ov 8\pi} \int_{\Sigma^2} \hat{g}^*\Tr(R\wedge R),
\end{equation}
which gives a 2-form on \( Y \) representing the first Chern class of the determinant line bundle over field space.

In fact, the local anomaly polynomial vanishes for the Cartan connections introduced in \eqref{eq:CartanC}. To see this, the curvature two-form can be obtained using \eqref{eq:MCw} and the Maurer-Cartan equation as
\beq
R^a_{\ b}&=(-\lam+\lam^2) (\theta_{\text{adj}})^a_{\ c} \wedge(\theta_{\text{adj}})^c_{\ b} .
\eeq
where $\theta_{\text{adj}}=\theta^a T^{\text{adj}}_a$ is the Maurer-Cartan one-form in the adjoint representation. The anomaly polynomial then can be shown to vanish using the cyclicity of the trace
\beq
\Tr(R \wedge R)=(-\lam+\lam^2)^2 \Tr(\theta_{\text{adj}}\wedge \theta_{\text{adj}}\wedge \theta_{\text{adj}}\wedge \theta_{\text{adj}})=0,
\eeq
which ensure that our sigma model of interest is free of local anomaly.

\subsection{Topology of the Configurtaion Space} \label{sec:top}

Next, we turn to the discussion of the global anomaly. It is associated with the nontrivial holonomy of the fermion path integral line bundle around closed loops in the configuration space \( \mathrm{Map}(\Sigma^2, G) \). A necessary condition for the presence of a global anomaly is that this configuration space admits noncontractible loops; in other words, the fundamental group \( \pi_1(\mathrm{Map}(\Sigma^2, G)) \) must be nontrivial. In what follows, we demonstrate that this is indeed the case under suitable conditions. (Readers primarily interested in the result may skip the proof.)

\begin{lemma}
Let \( \Sigma^2 \) be a connected two-dimensional manifold and let \( G \) be a connected, simply connected topological space. Then the fundamental group of the mapping space from \( \Sigma^2 \) to \( G \) is given by
\beq
\pi_1(\mathrm{Map}(\Sigma^2, G)) \cong H^3(\Sigma^2 \times S^1; \pi_3(G))=H^3(\Sigma^2 \times S^1; \mathbb Z).
\eeq

\end{lemma}

\begin{proof}
To compute the fundamental group of the mapping space \( \mathrm{Map}(\Sigma^2, G) \), we observe that loops in this space correspond to smooth one-parameter families of maps \( g_t: \Sigma^2 \to G \) with \( t \in S^1 \). This is equivalent to specifying a smooth map
\beq
\hat{g}: \Sigma^2 \times S^1 \to G.
\eeq
Therefore, the fundamental group is given by the homotopy classes of such maps:
\beq
\pi_1(\mathrm{Map}(\Sigma^2, G)) \cong [\Sigma^2 \times S^1, G].
\eeq

Now assume that \( G \) is connected and simply connected, so \( \pi_0(G) = \pi_1(G) = \pi_2(G) = 0 \). Since \( \Sigma^2 \times S^1 \) is a 3-manifold, only the first three homotopy groups of \( G \) are detected by maps from \( \Sigma^2 \times S^1 \). In particular, up to dimension three, \( G \) is homotopy equivalent to the Eilenberg--MacLane space \( K(\pi_3(G), 3) \). More precisely, this equivalence arises as the 3-stage truncation of the Postnikov tower for \( G \). Applying the representability of cohomology, we find
\beq
[\Sigma^2 \times S^1, G] \cong [\Sigma^2 \times S^1, K(\pi_3(G), 3)] \cong H^3(\Sigma^2 \times S^1; \pi_3(G)),
\eeq
as claimed.
\end{proof}

Remarkably, this identification already suggests a close connection between the global anomaly and the Wess--Zumino (WZ) term \cite{Wess:1971yu,Witten:1983ar} (see also \cite{Lee:2020ojw} for a recent discussion on Wess--Zumino term). The WZ term is defined by integrating the pullback of a canonical closed 3-form \( \omega_3={H\ov 2\pi} \in H^3(G; \mathbb{Z}) \) over a three-manifold \( M^3 \) equipped with a map \( g: M^3 \to G \):
\beq
\Gamma_{\mathrm{WZ}}[g; M^3] = 2\pi \int_{M^3} g^* \omega_3,\quad\text{with }
\omega_3 = \frac{1}{24\pi^2} \operatorname{Tr}\left[(g^{-1}dg)^3\right].
\eeq
Since \( g^* \omega_3 \) represents a class in \( H^3(M^3; \mathbb{Z}) \), the WZ term assigns an integer to each homotopy class of maps from \( M^3 \) into \( G \). In particular, when \( M^3 = \Sigma^2 \times S^1 \), a nonzero WZ term signals that the corresponding map \( g \) defines a nontrivial class in \( H^3(\Sigma^2 \times S^1; \mathbb{Z}) \), and hence corresponds to a noncontractible loop in the configuration space \( \mathrm{Map}(\Sigma^2, G) \). In the analysis that follows, we will concretely demonstrate how the global anomaly is realized as a nontrivial WZ term.

\subsection{The Global Anomaly Formula} \label{sec:hol}

To compute the global anomaly, we will use the general formula for the holonomy of the determinant line bundle, as given by the Witten's global anomaly formula \cite{Witten:1985xe} and their mathematical formulations \cite{Bismut:1986wr,Dai:1994kq}. The global anomaly of a Weyl fermion on a \( d = 2n \) dimensional manifold \( M^d \) with anti-periodic (bounding) spin structure\footnote{The global anomaly formula for the periodic spin structure along the loop differs from the anti-periodic case by a factor of \( (-1)^{\operatorname{ind}(D_+)} \)~\cite{Dai:1994kq}. In this work, we choose the anti-periodic spin structure, as it provides a natural setting for applying the APS index theorem, as will become clear in Section~\ref{sec:APS1}. } along a loop $C$ in the configuration space is expressed as
\beq \label{eq:WBF}
\text{hol}(C) =  \lim_{\ep\rightarrow 0}\exp(-2\pi i\, \bar{\eta}^\ep),
\eeq
where \( \operatorname{ind}(D_+) \) is the standard index of the chiral Dirac operator on \( M^d \), and \( \bar{\eta}^\epsilon \) is the \emph{reduced eta invariant} computed on the mapping torus \( M^d \times S^1 \) defined by the loop \( C \), with respect to a certain self-adjoint Dirac operator \( \slashed{D}_3^\epsilon \) in the adiabatic limit, where the metric is taken as \( g_{M^d} \oplus \epsilon^{-1} g_{S^1} \) (see~\cite{Bismut:1986wr} for more details). The reduced eta invariant \( \bar{\eta} \) associated to a Dirac operator \( D \) is defined by
\begin{equation}
\bar{\eta} = \frac{1}{2}(\eta + h), \quad \text{where} \quad 
\eta := \lim_{s \to 0^+} \sum_{\lambda \neq 0} \operatorname{sign}(\lambda)\, |\lambda|^{-s}, \quad 
h = \dim \ker(D).
\end{equation}

In our case of a Majorana--Weyl fermion, the anomaly is effectively given by the square root of~\eqref{eq:WBF}. As a result, the global anomaly takes the form \cite{Freed:1986hv}
\begin{equation} 
\text{hol}(C)_{\text{MW}} =\lim_{\ep\rightarrow 0}\exp(-\pi i\, \bar{\eta}^\ep).
\end{equation}

\subsection{Applicability of the APS Index Theorem.} \label{sec:APS1}

In general, computing the reduced eta invariant \( \bar{\eta} \) on a \((d+1)\)-dimensional mapping torus is a nontrivial task. However, when the mapping torus bounds a \((d+2)\)-dimensional manifold equipped with compatible geometric structures—such as a spin structure and a principal \(G\)-bundle—its bordism class is trivial in the appropriate spin bordism group, and one can apply the Atiyah–Patodi–Singer (APS) index theorem \cite{Atiyah:1975jf} to compute \( \bar{\eta} \) from data on the bounding manifold. In the case of a standard global anomaly, this is where the difficulty arises: the mapping torus generically fails to bound a spin manifold with compatible gauge bundle, and the APS theorem cannot be used. Remarkably, we will show that in the case of sigma model anomalies, we can bypass this difficulty, and the reduced eta invariant can indeed be computed via the APS index theorem. (We remark that computation of global anomalies using the APS index theorem appeared early in the literature,  notably in the context of global anomaly cancellation in the heterotic string; see~\cite{witten1985global}.)

To appreciate the subtle difference behind global sigma model anomalies, it is helpful to first recall the standard notion of global anomalies associated with a symmetry group \(G\). A global anomaly arises when the gauge transformation under consideration is not connected to the identity component of the gauge group. In such cases, the associated \((d+1)\)-dimensional mapping torus encodes a nontrivial twisting of the gauge bundle, and typically fails to bound any \((d+2)\)-dimensional spin manifold over which both the spin structure and the principal \(G\)-bundle extend. This failure to bound is precisely measured by a nontrivial element in the spin bordism group \( \Omega_{d+1}^{\mathrm{Spin}}(BG) \) \cite{Kapustin:2014tfa,Kapustin:2014dxa,Yonekura:2018ufj}, and signals the presence of a genuine global anomaly.

The distinction between local and global anomalies is directly reflected in whether the mapping torus can be smoothly extended over \( M^d \times D^2 \). In the case of a local anomaly, where the gauge transformation \(u: M^d \to G \) is connected to the identity, the mapping torus \( M^d \times_u S^1 \) is topologically trivial, and both the spin structure (assuming anti-periodicity along \( S^1 \)) and the gauge bundle extend smoothly over \( M^d \times D^2 \). In contrast, a global anomaly arises when \( u \) is homotopically nontrivial. In this case, the associated twist in the gauge bundle obstructs its extension over the disc, and the mapping torus fails to bound a spin manifold with compatible gauge data.

We now turn to the case of sigma model anomalies. Here, the relevant geometric data consists of a spin structure and a map \( g: M^d \times S^1 \to G \), where \(G\) is the target space of the sigma model. For global anomalies in this context, the mapping torus corresponds to a nontrivial element in \( \pi_1(\text{Map}(M^d, G)) \). 

The question of whether the APS index theorem applies naturally reduces to whether the mapping torus represents a trivial (i.e., null-bordant) element in the spin bordism group \( \Omega_{d+1}^{\mathrm{Spin}}(G) \). Again, the simplest possible candidate for a bounding \((d+2)\)-manifold is \( M^d \times D^2 \), but—as in the standard case of gauge-theoretic global anomalies—the mapping torus generally cannot be smoothly extended over \( M^d \times D^2 \) with compatible spin structure and map to \(G\).

Therefore, one might hastily conclude that computing \(\bar{\eta}\) for a global sigma model anomaly is just as difficult as in the standard global anomaly case. However, it is crucial to observe that the Dirac operator depends on the map \( g \) only through the pullback of a fixed connection on \( G \). 

Although the map \( g: M^d \times S^1 \to G \) may be homotopically nontrivial, the connection \( \mathcal A_\mu \) defined above lives in a trivial principal bundle. Explicitly, we have a principal bundle \( g^* F(G) \), where \( F(G) \) denotes the frame bundle of \( G \), and the associated bundle is \( g^* TG \). Crucially, \( F(G) \) is trivial since \( G \) is parallelizable—an important geometric property that ensures \( TG \cong G \times \mathbb{R}^n \) globally. This triviality descends to the pullback: \(\mathcal  A_\mu \) is a connection on the trivial principal bundle \( g^*F(G) \to M \), and is therefore always homotopically trivial.  Therefore the topological nontriviality of \(g\) enters only through the field configuration, but not through the bundle structure that supports the connection.

This leads to a key desired property: more generally, a connection defined on a trivial principal \(G\)-bundle over a manifold \( N \) can always be extended to any bounding manifold \( W \) with boundary \( \partial W= N \), precisely because the triviality of the bundle eliminates any topological obstruction. Concretely, a trivial principal bundle is globally isomorphic to \( N \times G \), so a connection \(\mathcal A \) on it can be described globally as a Lie algebra-valued 1-form \( A \in \Omega^1(N, \mathfrak{g}) \). To extend this to \( W \), one simply defines a smooth 1-form \( \mathcal A' \in \Omega^1(W, \mathfrak{g}) \) that restricts to \( \mathcal A \) on the boundary. Since there are no nontrivial transition functions or patching conditions to satisfy, the extension can always be carried out globally and smoothly. 

Now we are ready to apply this general discussion to our \( d = 2 \) setting, with \( N = \Sigma^2 \times S^1 \) and \( W = \Sigma^2 \times D^2 \), and to compute the eta invariant using the APS index theorem.
\subsection{Application of the APS Index Theorem} \label{sec:APS2}

To compute the reduced eta invariant \( \bar{\eta} \) on \( \Sigma^2 \times S^1 \), we use the Atiyah–Patodi–Singer index theorem by extending the geometric data to the bounding manifold \( \Sigma^2 \times D^2 \). In this section, we take the connection \( \mathcal{A} = g^* w \) to be a pull back of the Cartan connections \eqref{eq:CartanC}, where we have shown the absence of the local anomaly. 

 Since both the spin structure and the pulled-back connection \( \mathcal A \) extend smoothly over \( \Sigma^2 \times D^2 \), the APS index theorem applies and yields
\beq \label{eq:APS}
\bar\eta:= \frac{1}{2} (\eta + h)=  \int_{\Sigma^2 \times D^2} \hat{A}(R) \wedge \text{ch}(\mathcal F)-\text{ind}(\slashed{D}_4)  \,,
\eeq
where \( h = \dim \ker(\slashed{D}_3) \), \( \text{ind}(\slashed{D}_4) \) is the standard Dirac index on $\Sigma^2 \times D^2$ with the APS boundary conditions imposed on the boundary $\Sigma^2\times S^1$, \( \hat{A}(R) \) is the A-roof genus, and \( \text{ch}(\mathcal F) \) is the Chern character of the connection \(\mathcal A\).

We now turn to computing each term in \eqref{eq:APS} separately, beginning with the bulk integral.
On the product manifold \( \Sigma^2 \times D^2 \), the curvature tensor splits cleanly between the two factors due to the block-diagonal metric. As a result, the degree-4 component of \( \hat{A}(R) \wedge \text{ch}(\mathcal F) \) receives no contribution from mixed curvature terms. The APS bulk integral therefore simplifies to
\beq
\int_{\Sigma^2 \times D^2} \hat{A}(R) \wedge \text{ch}(\mathcal F) = \int_{\Sigma^2 \times D^2} \text{ch}_2(\mathcal F).
\eeq

In our setup, the connection \( \mathcal A \) is defined on a trivial principal \( G \)-bundle over the bulk manifold \( \Sigma^2 \times D^2 \). As a consequence, the second Chern character \( \operatorname{ch}_2(\mathcal F) ={1\ov2}\left({i\ov 2\pi}\right)^2\operatorname{Tr}(\mathcal F \wedge \mathcal F) =- \frac{1}{8\pi^2} \operatorname{Tr}(\mathcal F \wedge\mathcal  F) \) is an exact form, and this allows us to apply Stokes' theorem :
\beq \label{eq:Stokes}
\int_{\Sigma^2 \times D^2} \operatorname{ch}_2(\mathcal F) =- \int_{\Sigma^2 \times S^1} \frac{1}{8\pi^2} \text{CS}_3(\mathcal A).
\eeq
where $\text{CS}_3(\mathcal A) = \operatorname{Tr} \left( \mathcal  A \wedge d \mathcal A + \tfrac{2}{3} \mathcal A \wedge\mathcal A \wedge \mathcal A \right)$ is a Chern-Simons 3-form.

The Chern-Simons 3-form for generic Cartan connection using \eqref{eq:MCw} is given by
\beq
\text{CS}_3(\mathcal{A} = g^* w) =  \lambda^2 \left( -1 + \tfrac{2}{3} \lambda \right) \Tr_{\text{adj}}(g^{-1} dg)^3,
\eeq
where \( \operatorname{Tr} \) denotes the trace in the adjoint representation of \( G \).

Therefore the boundary integral in \eqref{eq:Stokes} proportional to the Wess-Zumino-Witten functional that captures the winding number of the map \( g \) from \(\Sigma^2\times S^1\) into \( G \) as
\beq
-\frac{1}{8\pi^2} \int_{\Sigma^2 \times S^1} \text{CS}_3(\mathcal A)
&= \frac{\lam^2(3-2\lam)}{24\pi^2} \int_{\Sigma^2 \times S^1} \Tr_{\text{adj}} \left( g^{-1} d g \right)^3
\\&= 2 \lam^2(3-2\lam) h^\vee \cdot {1\ov 24\pi^2} \int_{\Sigma^2 \times S^1} \Tr\left( g^{-1} d g \right)^3
\\&= 2\lam^2(3-2\lam) h^\vee\Gamma_{\mathrm{WZW}}[g;\Sigma^2\times S^1],
\eeq
where in the last line, \( h^\vee \) is the dual Coxeter number of $G$, appearing when the adjoint trace $\Tr_{\text{adj}}$ is rewritten in terms of the fundamental trace \( \Tr \), which sets the standard WZW normalization \cite{Witten:1983ar}.



Next, we discuss the \( \operatorname{ind}(\slashed{D}_4) \) term. We claim that \( \operatorname{ind}(\slashed{D}_4) \) is always even for any Dirac operator \( \slashed{D}_4 \) on an orientable 4-manifold, possibly twisted by a real vector bundle. This follows from the existence of an anti-linear symmetry \( \mathcal{C}_4 = C_4 \circ K \), where \( C_4=i\sigma_2\otimes\sigma_1 \) is the charge conjugation matrix in 4d and \( K \) denotes complex conjugation. We have
\beq
\mathcal{C} _4i\slashed{D}_4 \mathcal{C}_4^{-1} = i\slashed{D}_4 ,\quad \mathcal{C}_4^2 = -1,
\eeq
which ensures spectrum is doubly degenerate. Moreover, \( [\mathcal{C}_4,\gamma_c]=0 \) and hence the zero modes of definite chirality are also preserved under $\mathcal{C}_4$. Altogether, this guarantees that the difference in the number of positive and negative chirality zero modes, i.e., the index of \( \slashed{D}_4 \), is always an even integer:\footnote{%
We remark that for pseudo-real representations, one can define a modified anti-linear symmetry \( \mathcal{C}' _4= C_4 \otimes J \circ K \), where \( J \) is the intertwiner satisfying \( R(g)^* = J R(g) J^{-1} \). However, in this case \( {\mathcal{C}'_4}^2 = +1 \), so the space of zero modes is real rather than symplectic, and there is no enforced degeneracy. Consequently, the index need not be even.%
\label {ft:pr}}
\beq
\operatorname{ind}(\slashed{D}_4) \in 2\mathbb{Z}.
\eeq

Finally, gathering the preceding discussions, we arrive at the expression for the anomaly phase associated with a Majorana--Weyl fermion which takes the form
\beq \label{eq:globalanomaly}
\exp(-\pi i \bar{\eta}) = \exp\left( -\lam^2(3-2\lam) h^\vee \cdot 2 \pi i\,  \, \Gamma_{\mathrm{WZW}}[g;\Sigma^2\times S^1] \right),
\eeq
which is a generally a complex-valued phase, reflecting the complex nature of the Majorana-Weyl determinant and realizing it geometrically via the topological Wess-Zumino term of level
\begin{equation}
k_{\mathrm{anom}}(\lambda) \;=\; -\lambda^{2}\,\bigl(3-2\lambda\bigr)\,h^\vee,
\end{equation}
which can take a generic real value as $\lambda$ varies. This shows that the anomaly takes values in a generic element of U(1), varying continuously with the choice of connection on the target space.

For the three special values of the Cartan connection parameter $\lambda$ we have
\begin{equation}
k_{\mathrm{anom}}(\lambda) =
\begin{cases}
0, & \lambda = 0, \\[6pt]
-\dfrac{1}{2}\,h^\vee, & \lambda = \dfrac{1}{2}, \\[10pt]
-\,h^\vee, & \lambda = 1.
\end{cases}
\end{equation}
Note that for left and right Cartan connections $\lam=0,1$, the $k_{\mathrm{eff}}$ is integer and hence the global anomaly is absent. This is physically consistent, since in this case the fermion is completely decoupled from the sigma model due to the parallelizability of the target. For $\lambda = 0$, this is evident from the expression~\eqref{eq:MCw}, while for $\lambda = 1$ the spin connection vanishes in the (co)frame given by the right Maurer--Cartan 1-form $\tilde{\theta} = dg\,g^{-1}$.

\section{Anomalous Quantization of the WZ term} \label{sec:quan}

\subsection{Fermions in a Sigma Model Background} \label{sec:nondyn}

Here we treat sigma model background as a non-dynamical background with a chosen Cartan connection. The global anomaly \eqref{eq:globalanomaly} implies that for non-integer $k_{\mathrm{anom}}(\lambda)$, the fermionic path integral defines a nontrivial section of a line bundle over the space of sigma model background. 

We claim that the anomaly can be cancelled if one add a level-$k$ Wess--Zumino term \( k\, \Gamma_{\mathrm{WZ}}[g; M^3] \), where \( \partial M^3 = \Sigma^2 \), provided the following anomalous quantization condition on $k$ is satisfied:
\begin{equation} \label{eq:quan}
\boxed{\,k + k_{\mathrm{anom}}(\lambda) \;\in\; \mathbb{Z}\,,\quad \text{Majorana-Weyl fermions}}
\end{equation}

This can be proved as follows. Given a smooth homotopy \( g_t: \Sigma^2 \to G \) for \( t \in [0,1] \), and a corresponding smooth family of extensions \( \tilde{g}_t: M^3_{(t)}  \to G \), one may construct a 4-manifold \( M^3_{(t)} \times [0,1] \) whose boundary is 
\beq
\partial(M^3_{(t)} \times [0,1]) = \Sigma^2 \times [0,1] \cup M^3_{(1)}  \cup (-M^3_{(0)} ).
\eeq
Since the canonical 3-form \( \omega_3 = {1\ov 24\pi^2}\operatorname{Tr}[(g^{-1}dg)^3] \) on \( G \) is closed, \( d\omega_3 = 0 \), Stokes' theorem applied to this 4-manifold implies that the difference in Wess--Zumino terms is captured entirely by an integral over the interpolating cylinder \( [0,1] \times \Sigma^2 \). Explicitly, one obtains:
\beq \label{eq:WZcontribution}
\Gamma_{\text{WZ}}[\tilde g_1;M^3_{(1)} ] - \Gamma_{\text{WZ}}[\tilde g_0;M^3_{(0)} ] = \Gamma_{\text{WZ}}[g_t;\Sigma^2\times S^1].
\eeq
Therefore WZ term's contribution \eqref{eq:WZcontribution} precisely cancels that global anomaly \eqref{eq:globalanomaly} when obeys the quantization condition \eqref{eq:quan}.

A couple of remarks are in order. First, it is important to emphasize that the global anomaly arises not only for Majorana--Weyl fermion but also for Weyl fermions more generally, since the anomaly phase can take values in an arbitrary $U(1)$. In the case of Weyl fermions, the global anomaly is doubled and the quantization condition is simply modified as
\begin{equation} \label{eq:Wquan}
\boxed{\,k +2 k_{\mathrm{anom}}(\lambda) \;\in\; \mathbb{Z}\,,\quad \text{Weyl fermions}}
\end{equation}

Second, the anomalous quantization condition~\eqref{eq:quan} may be viewed as a two-dimensional analogue of the three-dimensional parity anomaly. However, it is richer: whereas the $3d$ parity anomaly forces the Chern--Simons level to be at worst a half-integer, in our case the level $k$ can take arbitrary real values due to the additional freedom in choosing connections on~$G$. This is a feature with no counterpart in ordinary gauge theory.

\subsection{$(1,0)$ Supersymmetric Sigma Model and UV--anchoring}  \label{sec:(1,0)}

Interestingly, a subtle difference arises when the sigma model field is taken to be dynamical, which distinct to be a unique feature of sigma model anomalies where do direct analog is known in gauge anomalies. Here we consider the model of our interest~\cite{Gaiotto:2019asa}, namely the $(1,0)$ supersymmetric sigma model with target space $G$ and Wess--Zumino coupling $k$, which exhibits a non-trivial RG flow. 

For a generic target space $M$, the $(1,0)$ supersymmetric sigma model on $\mathbb{R}^{1,1}$ can be written as~\cite{Hull:1985jv,Hull:1986xn}
\beq \label{eq:(1,0)action}
S
&={1\ov 4\pi \al'}\int d^2 x\, (g_{ij}(\vp)\eta^{\mu\nu}+b_{ij}(\vp) \epsilon^{\mu\nu} )\pa_\mu \vp^i \pa_\nu\vp^j+ig_{ij}(\vp) \psi^i(\pa_{-} \psi^j+\Gamma^j_{kl} \pa_{-}\vp^k \psi^l)
\eeq
where $\psi^i$ are sections of Majorana-Weyl spinors w.r.t. local coordinates $\{\vp^i\}$. In terms of $B={1\ov 2}b_{ij} d\vp^i \wedge d\vp^j$ and $H=dB={1\ov 3!} H_{ijk}d\vp^i \wedge d\vp^j \wedge d\vp^k$, the connection $\Gamma^j_{kl}$ is written as
\beq
\Gamma^i_{jk}=\Christoffel{i}{j}{k} -S^i_{jk},
\eeq
where $\Christoffel{i}{j}{k}$ are the Levi-Civita connection on $TG$ and $S^i_{jk}={1\ov 2}g^{il}H_{jkl}$ is the torsion part of the affine connection. 

Now we focus on to our group manifold case. Here, we have
\beq
H_{ijk}=k\al' T(F)  \,f_{cde} \theta^c_i \theta^d_j \theta^e_k ,
\eeq
with $T(F)$ is the Dynkin index of the fundamental representation and the expression comes from the identification
\beq
{1\ov 4\pi \al'} \int \vp^* B={1\ov 4\pi \al'} \int_{M^3} \vp^* H =k \int_{M^3} \vp^*{1\ov 12\pi}\Tr(g^{-1}dg)^3.
\eeq
If we translate $\Gamma^i_{jk}$ to the spin-connection in the Maurer-Cartan basis, it is straightforward to see that the connection is nothing but the Cartan connections introduced in \eqref{eq:CartanC}.

Under the overall scaling of the metric $g_{ij}\rightarrow \lam g_{ij}$, we note that the torsionless and torion part of the connection behave as
\beq
\Christoffel{i}{j}{k} \rightarrow \Christoffel{i}{j}{k},\quad S^i_{jk} \rightarrow \lam^{-1} S^i_{jk}.
\eeq

Now it is important to understand the RG properties of the $(1,0)$ action \eqref{eq:(1,0)action}. The sigma model has logarithmic divergences at one-loop and it gives a beta-function (which is essentailly same as bosonic sigma model case) given by \cite{Curtright:1984dz,Sen:1985eb,Sen:1985qt}
\beq
~&\mu {d\ov d\mu}g_{ij}=\alpha' \tilde R_{(ij)} + O(\alpha'^2)=\alpha'\Big(R_{ij}-\tfrac14 H_{ik\ell}H_j{}^{k\ell}\Big)+ O(\alpha'^2),
\\& \mu\frac{d}{d\mu} B_{ij}
= \alpha' \tilde R_{[ij]}+ O(\alpha'^2)=\alpha' \Big(-\tfrac12 \nabla^k H_{kij}\Big) + O(\alpha'^2).
\eeq
where $\tilde R_{ij}$ is a Ricci tensor computed w.r.t. torsionful connection $\Gamma^i_{jk}$. In the case of a group manifold, there is no running of $B_{ij}$ at one loop, since $\nabla^k H_{kij}=0$. Moreover, the two-loop contribution also vanishes in both the bosonic and supersymmetric cases \cite{Braaten:1985is,Fridling:1985hc}. In fact, for the bosonic PCM+WZ model the vanishing of the $\beta$--function of $B_{ij}$ in the presence of covariantly constant torsion $H_{ijk}$ (which indeed holds for $G$) was established to all orders in perturbation theory in \cite{Mukhi:1985vy}. It is then straightforward to extend this non-renormalization theorem to the WZ term of the $(1,0)$ supersymmetric model using superspace techniques. We therefore focus on the single running coupling in the metric sector.

Using the MC basis $\{\theta^a\}$, on a compact simple group $G$ with bi--invariant metric 
$g_{ab}=r^2\delta_{ab}$ and $H_{abc}=h\,f_{abc}$, 
the torsionful connection is
\beq
\Gamma^{a}{}_{bc}
=\tfrac12 f^{a}{}_{bc}-\tfrac12 H^{a}{}_{bc}
=\Big[\tfrac12\Big(1-\frac{h}{r^2}\Big)\Big] f^{a}{}_{bc}.
\eeq
If we parameterize the Cartan connections by 
$\nabla_X Y=\lambda [X,Y]$, then
\beq
\lambda(\mu)=\tfrac12\Big(1-\frac{h}{r^2(\mu)}\Big), 
\qquad h=\alpha' k.
\eeq
Using the one--loop flow,
\beq
\mu\frac{d}{d\mu}r^2
=\frac{\alpha' h^\vee}{2}\Big(1-\frac{(\alpha'k)^2}{r^4}\Big),
\eeq
we can express $\lambda(\mu)$ implicitly as follows.  Writing $y(\mu)=r^2(\mu)$,
\beq
y-\frac{\alpha'k}{2}\ln\!\frac{y+\alpha'k}{y-\alpha'k}
=\Big[y_0-\frac{\alpha'k}{2}\ln\!\frac{y_0+\alpha'k}{y_0-\alpha'k}\Big]
+\frac{\alpha' h^\vee}{2}\ln\!\frac{\mu}{\mu_0},
\eeq
and then
\beq
\lambda(\mu)=\tfrac12\left(1-\frac{\alpha'k}{y(\mu)}\right).
\eeq
In the UV ($r^2\gg \alpha'k$), we find
\beq
r^2(\mu)\;\simeq\;\frac{\alpha' h^\vee}{2}\,\ln\frac{\mu}{\Lambda},
\qquad
\lambda(\mu)\;\simeq\;\tfrac12-\frac{k}{h^\vee\,\ln(\mu/\Lambda)} 
+ O\!\big((\ln\mu)^{-2}\big),
\eeq
which shows an asymptotic freedom. In particular, the connection becomes the Levi-Civita connection at the UV fixed point as $\lam(\infty)={1\ov 2}$.

Here we see a distinct feature of sigma--model anomalies, where the connection undergoes a non--trivial RG flow upon integrating out the sigma--model fields from UV to IR. This is in stark contrast with standard gauge anomalies, where fermions always couple minimally to the gauge field regardless of the energy scale.

However, this leads to an apparent puzzle: the global sigma--model anomaly seems RG--sensitive, since $\lambda(\mu)$ has nontrivial scale dependence and, according to \eqref{eq:globalanomaly}, the global anomaly depends explicitly on $\lambda(\mu)$ and hence on the RG scale $\mu$.

We propose that the subtle resolution of the puzzle is that, once the sigma--model field is dynamical, the anomaly expressed in terms of the determinant of the Dirac operator must be computed only in the deep UV, a prescription we refer to as \emph{UV--anchoring}. To justify this, we adopt the Wilsonian picture of RG flow \cite{Wilson:1971bg,Wilson:1973jj}. One defines the theory with a UV cutoff $\Lambda$ and then lowers the sliding cutoff $\mu<\Lambda$ by integrating out shells of high--momentum modes. In the process, an infinite tower of higher--dimensional irrelevant operators consistent with the symmetries is generated at finite $\mu$. In our case this includes higher--dimensional operators involving fermions. Their presence spoils the interpretation of the fermionic path integral at scale $\mu$ as the determinant of a Dirac operator, rendering the usual global anomaly formula in terms of the $\eta$--invariant inapplicable at finite $\mu$. The correct, RG--invariant definition of the anomaly is therefore obtained by evaluating the fermion determinant in the UV, before such terms are generated.


From this UV--anchoring prescription, the global sigma--model anomaly of the
$(1,0)$ theory is computed using~\eqref{eq:globalanomaly} with the UV value of
the Cartan connection. One finds
\begin{equation}
  k_{\mathrm{anom}}\!\left(\lambda(\infty)\right) \;=\; -\,\frac{h^\vee}{2}\,,
\end{equation}
so that the anomalous quantization condition for the $(1,0)$ sigma model reads
\begin{equation}
  k - \frac{h^\vee}{2} \;\in\; \mathbb{Z}\,.
\end{equation}
This matches precisely the condition obtained independently in
\cite{Gaiotto:2019asa}, where it arose from the `t~Hooft anomaly of the global
$G_L\times G_R$ symmetry.

\section{Absence of the Dai-Freed anomaly} \label{sec:DF}

The anomalous quantization condition discussed in Section~\ref{sec:quan}, which cancels the global anomaly, is not by itself sufficient to ensure a fully well-defined fermionic path integral. Rather, it only guarantees that the phase of the path integral is well defined on a fixed two-dimensional manifold~$\Sigma^2$.

However, to be consistent with locality---in the sense that cutting and pasting of space-time manifolds should be well defined---the phase of the fermionic path integral must be defined consistently on arbitrary two-dimensional manifolds (orientable, in our case). The failure to assign an absolute meaning to the phase of the fermionic path integral is known as the Dai--Freed anomaly~\cite{Dai:1994kq,Witten:2015aba,Witten:2016cio,Yonekura:2016wuc,Garcia-Etxebarria:2018ajm,Witten:2019bou}. 

Following the spirit of~\cite{Witten:2019bou}, we propose that the unambiguous definition of the phase of the Majorana--Weyl fermionic path integral on an arbitrary two-dimensional manifold~$\Sigma^2$, which bounds a three-manifold~$M^3$ where all relevant structures extend smoothly, is
\begin{equation} \label{eq:absphase}
  Z_{\Sigma^2} \;=\; |Z_{\Sigma^2}| \,
   \exp\!\left(- i \pi \, \bar{\eta}_{M^3}\right) \,
   \exp\!\left(2 \pi i\, k \, \Gamma_{\mathrm{WZW}}[g;M^3]\right),
\end{equation}
where $k$ satisfies the anomalous quantization condition
\beq
k + k_{\mathrm{anom}}(\lambda) \;\in\; \mathbb{Z},
\eeq
for the non-dynamical background, and in particular is fixed by the UV value of the connection in the case of a dynamical background. The first phase factor in~\eqref{eq:absphase}, proportional to $\bar\eta$, arises from the Dai--Freed theorem, while the second factor, proportional to the WZW term, cancels the global sigma model anomaly discussed in Section~\ref{sec:2}.

To establish that the phase in~\eqref{eq:absphase} is unambiguous for an arbitrary surface~$\Sigma^2$ admitting a spin structure and a bounding three-manifold~$M^3$, we must show that the definition is independent of the choice of extension to the $(d{+}1)=3$-dimensional manifold~$M^3$. The necessary and sufficient condition is that, for any \emph{closed} spin three-manifold~$\bar M^3$, one has
\beq \label{eq:DFcon}
\exp\!\left(-i\pi \, \bar \eta_{\bar M^3}\right)\,
\exp\!\left(2\pi i\, k \, \Gamma_{\mathrm{WZW}}[g;\bar M^3]\right)
\;\stackrel{?}{=}\; 1 .
\eeq

We now show that the condition~\eqref{eq:DFcon} indeed holds. Following the spirit of Section~\ref{sec:APS1}, we compute the eta-invariant using the APS index theorem. As discussed in Section~\ref{sec:APS1}, one has $\Omega_3^{\mathrm{Spin}}(G) = \mathbb{Z}$, with the obstruction arising from the extension of topologically non-trivial sigma model maps, since $\pi_3(G) = \mathbb{Z}$. However, because the fermions couple to the sigma model maps only through the pullback of the connection on $TG$, it suffices to extend the connection itself through some bounding four-manifold. Using the facts that $\Omega_3^{\mathrm{Spin}}(pt)=0$ and that $TG$ is trivial, we conclude that for any spin three-manifold there always exists a four-manifold on which such a connection can be smoothly extended.

The application of the APS index theorem proceeds in close analogy with Section~\ref{sec:APS2}, and the result is simply a generalized version of~\eqref{eq:globalanomaly}, namely
\beq 
\exp\!\left(-\pi i \, \bar{\eta}_{\bar M^3}\right) 
= \exp\!\left(-2\pi i\, k_{\text{anom}} \, \Gamma_{\mathrm{WZW}}[g;\bar M^3]\right).
\eeq
This relation makes it manifest that the consistency condition~\eqref{eq:DFcon} is satisfied, and hence the theory is free from any Dai--Freed anomaly.

In our case, this is not the end of the story: not every spin surface~$\Sigma^2$ equipped with a sigma–model map $g:\Sigma^2 \to G$ admits a bounding three–manifold to which all structures extend. The obstruction is purely \emph{spin}—namely, the Arf invariant—since $\Omega^{\mathrm{Spin}}_2(pt)=\mathbb{Z}_2$ while $\pi_2(G)=0$. Consequently, if $\operatorname{Arf}(\Sigma^2)=1$, there is no spin three–manifold $M^3$ with $\partial M^3=\Sigma^2$ and extension $\tilde g:M^3\to G$; in particular, the absolute phase prescription~\eqref{eq:absphase} cannot be used to define a canonical phase of the path integral on such~$\Sigma^2$.
 
This does not imply that the phase cannot be consistently defined on arbitrary~$\Sigma^2$; rather, it means that there are multiple consistent assignments of the phase, which can therefore be regarded as a free parameter of the theory~\cite{Witten:2016cio,Witten:2019bou}. In our case, the obstruction is $\mathbb Z_2$-valued, so any disjoint union of two copies of~$\Sigma^2$ bounds a three–manifold to which all structures extend smoothly. Consequently, for $\Sigma^2$ with nontrivial Arf invariant, only the squared partition function
\beq
Z_{\Sigma^2 \sqcup \Sigma^2} \;=\; Z_{\Sigma^2}^2
\eeq
is uniquely defined. This leaves two consistent choices for the phase of $Z_{\Sigma^2}$ itself,
\beq
Z_{\Sigma^2} \;=\; \pm \sqrt{\,Z_{\Sigma^2}^2}\,.
\eeq

\begin{appendix}

\end{appendix}

\bibliographystyle{utphys}
\bibliography{anomalies}

\end{document}